\documentclass[letterpaper, 10 pt, journal]{IEEEtran}
\IEEEoverridecommandlockouts

\usepackage{amsmath,dsfont,amssymb}
\usepackage{amsthm}
\usepackage{graphicx}
\usepackage{amsfonts}
\usepackage{float}
\usepackage{hybridsystem}

\usepackage[colorlinks = true,
            linkcolor = black,
            urlcolor  = blue,
            citecolor = blue,
            anchorcolor = black]{hyperref}
\usepackage{algorithmicx}
\usepackage{algorithm}
\usepackage{algpseudocode}
\usepackage{lipsum}
\usepackage{subfig}
\usepackage{amssymb}
\usepackage{float}
\usepackage{balance}
\usepackage{color}
\usepackage{mathtools}
\theoremstyle{plain}
\newtheorem{thm}{\textbf{Theorem}}

\newtheorem{prop}{\textbf{Proposition}}

\theoremstyle{definition}

\newtheorem{ass}{\textbf{Assumption}}
\theoremstyle{remark}
\newtheorem{rem}{\textbf{Remark}}

\allowdisplaybreaks
\usepackage{algpseudocode}

\usepackage[normalem]{ulem}

\begin{document}

\title{\LARGE \bf
Online decentralized decision making with inequality constraints: an ADMM approach
}
\author{Yuxiao Chen$^{1}$, Mario Santillo$^{2}$, Mrdjan Jankovic$^{2}$,
and Aaron D. Ames$^{1}$
\thanks{$^{1}$ Yuxiao Chen and Aaron D. Ames are with Department of Mechanical and Civil Engineering, California Institute of Technology, Pasadena, CA, USA {\tt\small chenyx,ames@caltech.edu}}
\thanks{$^{2}$ Mario Santillo and Mrdjan Jankovic are with Ford Research and Advanced Engineering, Dearborn, MI, USA {\tt\small msantil3,mjankov1@ford.com}}

}

\maketitle
\thispagestyle{empty}

\begin{abstract}
We discuss an online decentralized decision making problem where the agents are coupled with affine inequality constraints. Alternating Direction Method of Multipliers (ADMM) is used as the computation engine and we discuss the convergence of the algorithm in an online setting. To be specific, when decisions have to be made sequentially with a fixed time step, there might not be enough time for the ADMM to converge before the scenario changes and the decision needs to be updated. In this case, a suboptimal solution is employed and we analyze the optimality gap given the convergence condition. Moreover, in many cases, the decision making problem changes gradually over time. We propose a warm-start scheme to accelerate the convergence of ADMM and analyze the benefit of the warm-start. The proposed method is demonstrated in a decentralized multiagent control barrier function problem with simulation.
\end{abstract}
\begin{IEEEkeywords}
Decentralized control; ADMM; Control barrier Functions
\end{IEEEkeywords}

\section{Introduction}\label{sec:intro}
Multiagent systems and networked systems make up a large portion of the autonomous applications, such as multi-robot systems, autonomous vehicles, and power networks. In many cases, due to the large number of agents and the distributed nature, centralized decision making is not implementable and people resort to decentralized algorithms. These decentralized decision making problems are usually formulated as distributed optimization problems, which have been studied for decades \cite{tsitsiklis1986distributed,nedic2014distributed}. One typical setup is to optimize over a summation of functions depending on different subsets of the decision variables via local optimization and communication between agents. The algorithm structure and communication protocol largely depend on the coupling graph topology, including the master-worker setting \cite{agarwal2011distributed}, fully connected setting, and fully distributed setting. Time varying topologies and time delays have also been considered in the literature \cite{nedic2014distributed,yang2016distributed}.


Alternating Direction Method of Multipliers (ADMM) is one class of optimization algorithms that receive increasing popularity as a simple yet effective framework for distributed optimization \cite{boyd2011distributed}. Comparing to the gradient based methods such as \cite{nedic2014distributed}, it is based on dual ascent, which solves the primal problem locally, and uses the dual variable to coordinate the local optimizations at each node. A common way of applying ADMM to distributed optimization is to formulate it as a consensus problem, where each agent is associated with a local objective function, and are coupled with other agents in the system by a consensus constraint. The idea is to keep local copies of variables shared among multiple objective functions and enforcing consistency of the local copies, which fits right into the ADMM formulation. There exist a plethora of different variations of ADMM \cite{chen2018validating} with different communication topologies and convergence guarantees \cite{shi2014linear,zhang2014asynchronous}.

In this paper we consider the problem of decentralized decision making with multiple agents coupled by linear inequality constraints. The goal is to search for the solution that minimizes the violation. Examples of such problems include motion planning for multiple agents \cite{borrmann2015control,chen2020guaranteed,murray2007recent}, decentralized model predictive control \cite{bemporad2011decentralized}, and decentralized coordination for power grids \cite{li2016decentralized,fang2019decentralized}. Moreover, we consider the case where the decentralized decision making problem is solved repeatedly with gradually changing data, i.e., the inequality constraints and local objective functions changes gradually between successive time steps. In online decision-making problems, the decision need to be made within a limited amount of time, which can be several seconds, or even milliseconds, depending on the application. The motivating example is the obstacle avoidance problem for autonomous vehicles using control barrier functions, where the actions need to be determined in the order of 100 milliseconds. The main contributions of this paper are (1) we propose a consensus-based framework for online decentralized decision making with coupling inequality constraints (2) we apply the Prox-JADMM algorithm with warm-start for online decentralized decision making and analyzed the optimality gap for unconverged solutions and the continuity of the optimal solution under changing problem data.

\newsec{Paper structure.} Section \ref{sec:setup} describes the problem setup, Section \ref{sec:ADMM} reviews the ADMM setup and the convergence result, and Section \ref{sec:main} presents the main result, the warm-start Prox-JADMM algorithm for online decentralized decision making. We show the application of the proposed algorithm on a autonomous vehicle obstacle avoidance example in Section \ref{sec:result} and finally we conclude in Section \ref{sec:conclusion}.
\vspace{-0.2cm}
\section{Preliminaries and problem setup}\label{sec:setup}

We begin with some definitions. A function $f:\mathcal{X}\to\mathbb{R}$ is convex if $\forall x_1,x_2\in \mathcal{X}$, $\lambda\in [0,1]$, $f(\lambda x_1+(1-\lambda)x_2)\le\lambda f(x_1)+(1-\lambda)f(x_2)$. Given a convex function $f$, if $f$ is differentiable at $x$, then $\nabla f(x)$ denote the gradient of $f$ at $x$; $\partial f(x)$ denotes the subgradient of $f$, which is defined as
\begin{equation*}
  \partial f(x)\coloneqq \{c\in\mathbb{R}^n\mid\forall x_1\in\mathcal{X},f(x_1)-f(x)\ge c^\intercal(x_1-x)\}.
\end{equation*}
When $f$ is differentiable at $x$, $\partial f=\{\nabla f(x)\}$. A differentiable $f$ is strongly convex with parameter $\sigma$ if $\forall x,y\in\mathcal{X}, f(y)-f(x)\ge \nabla f(x)^\intercal(y-x)+\frac{\sigma}{2}||x-y||^2$. $|\cdot|$ denotes the element-wise absolute value of a vector or a matrix, $||\cdot||$ denotes the 2-norm of a vector, $||\cdot||_F$ denotes the Frobenius norm.

We consider a fully decentralized setting with $N$ agents, where $N$ may vary over time and the algorithm does not depend on $N$. Let $x_i\in\mathbb{R}^{n_i}$ be the decision variable for the $i$th agent, each agent is associated with the local objective function $f_i:\mathbb{R}^{n_i}\to\mathbb{R}$. A coupling linear inequality is a 3-tuple $(s,A_s,b_s)$ where $s\subseteq\{1,2,...,N\}$ is a subset of agents, $A_s$ and $b_s$ defines the following coupling constraint:
\begin{equation*}
  \sum\nolimits_{i\in s}A_s^i x_i\le b_s,
\end{equation*}
where $A_s^i$ is the block of $A_s$ corresponding to $x_i$. Two agents $x_i$ and $x_j$ are neighbors of each other if there exists a coupling constraint $(s,A_s,b_s)$ where $s$ contains both $i$ and $j$, and $\mathcal{N}_i$ denotes the neighbor set of agent $i$, with cardinality $N_i$. Let $\mathcal{C}$ denote the set of all coupling constraints, the optimization problem we consider is then:
\begin{equation}\label{eq:setup}
\begin{aligned}
  \mathop{\min}\limits_{x_1,...,x_N}& \sum\nolimits_{i} {f_i(x_i)}\\
  \mathrm{s.t.}&~x_i\in\mathcal{X}_i,~\forall (s,A_s,b_s)\in \mathcal{C},~\sum\nolimits_{i\in s}A_s^i x_i\le b_s,
\end{aligned}
\end{equation}
where $\mathcal{X}_i$ is the domain of $x_i$, assumed to be convex with a nonempty interior. The compact form of \eqref{eq:setup} is then
\begin{equation}\label{eq:compact_setup}
  \begin{aligned}
  \mathop{\min}\limits_{x_1,...,x_N}& \sum\nolimits_{i} {f_i(x_i)}\\
  \mathrm{s.t.}&~x_i\in \mathcal{X}_i,~Ax\le b,
\end{aligned}
\end{equation}
where $x=\begin{bmatrix}x_1^\intercal &  \cdots  & x_N^\intercal\end{bmatrix}^{\intercal}$ is the vector consisting of all $x_i$, $A$ and $b$ matrices are obtained by stacking the $A_s$ and $b_s$ matrices on the proper dimensions. In practice, \eqref{eq:setup} can be infeasible and we solve the following relaxed problem instead:
\begin{equation}\label{eq:relaxed_setup_cen}
\resizebox{0.9\hsize}{!}{$
  \mathop{\min}\limits_{x_1\in\mathcal{X}_1,...,x_N\in\mathcal{X}_N} \sum\nolimits_{i} {f_i(x_i)}+\beta\sum\limits_{(s,A_s,b_s)\in \mathcal{C}}\mathds{1}^\intercal\max\{\mathbf{0},\sum\limits_{i\in s}A_s^i x_i-b_s\}
$}
\end{equation}
where the maximum is taken entry-wise, $\beta>0$ is the penalty on constraint violation, and $\mathds{1}$ is a vector of all ones. The penalty term is a piecewise affine convex function of the $\{x_i\}$.

\begin{rem}
  In \cite{wang2011control}, the inequality constraint is enforced with a logarithmic penalty. However, since we consider the cases where \eqref{eq:compact_setup} may be infeasible, a linear penalty is chosen instead. Assuming $\forall x,c_1\le\sum\nolimits_{i} {f_i(x_i)}\le c_2$, if \eqref{eq:compact_setup} is feasible, then the solution to \eqref{eq:relaxed_setup_cen} satisfies $\mathds{1}^\intercal\max\{\mathbf{0},\sum\limits_{i\in s}A_s^i x_i-b_s\}\le \frac{c_2-c_1}{\beta}$. This means constraint violation due to relaxation can be made arbitrarily small by increasing $\beta$.
\end{rem}

To solve \eqref{eq:relaxed_setup_cen} decentrally, the local copy idea is adopted. To be specific, each agent keeps local copies of the decision variables of its neighbors, and additional consensus constraint is added so that these local copies agree with the actual variable. Let $x_j^i$ denote the local copy of $x_j$ kept by agent $i$, $\mathbf{x}_i\coloneqq [x_i^\intercal,{x_{j_1}^i}^\intercal,...,{x^i_{j_{N_i}}}^\intercal]^\intercal\in\mathbb{R}^{\mathbf{n}_i}$ be the vector consisting of $x_i$ and all the local copies of agent $i$'s neighbors, and let $\mathbf{A}^i\mathbf{x}_i\le\mathbf{b}^i$ be the compact form of the constraint
\begin{equation*}
  \forall (s,A_s,b_s)\in \mathcal{C},i\in s, A_s^i x_i+\sum\nolimits_{j\in s,j\neq i}A_s^j x_j^i\le b_s,
\end{equation*}
$\mathbf{A}^i \in\mathbb{R}^{m_i\times\mathbf{n}_i}$ and $\mathbf{b}^i\in\mathbb{R}^{\mathbf{n}_i}$. Then define the augmented local objective function as
\begin{equation}\label{eq:aug_obj}
  F_i(\mathbf{x}_i)= {f_i(x_i)}  +\beta (w^i)^\intercal \max\{\mathbf{0},\mathbf{A}^i \mathbf{x}_i-\mathbf{b}^i\},
\end{equation}
where $w^i\in\mathbb{R}^{m_i}$ is a vector with $w^i_j=1/|s_j|$, $s_j$ being the subset of agents corresponding to the $j$th row of $\mathbf{A}^i$, and $|s_j|$ being its cardinality. Since the constraint is shared by $|s_j|$ agents, each local objective function gets $1/|s_j|$ of the total penalty.
The decentralized optimization problem is then:
\begin{equation}\label{eq:relaxed_setup}
\begin{aligned}
  \mathop{\min}\limits_{\mathbf{x}_1,...,\mathbf{x}_N}& \sum\nolimits_{i} F_i(\mathbf{x}_i)\\
  \mathrm{s.t.}&~\forall i, x_i\in\mathcal{X}_i,~\forall j\in\mathcal{N}_i,x_j^i=x_j,
\end{aligned}
\end{equation}
which shall be solved with ADMM.
\section{Review of ADMM}\label{sec:ADMM}
We briefly review the ADMM framework.

\newsec{Standard ADMM.} The standard ADMM algorithm considers the following problem:
\begin{equation}\label{eq:admm}
  \begin{aligned}
  \mathop {\min }\limits_{x,z} &f(x) + g(z)  \\
  \mathrm{s.t.}&\;Ax + Bz - c = 0,
\end{aligned}
\end{equation}
where $x\in\mathbb{R}^n,z\in\mathbb{R}^m$ are the decision variables, $f:\mathbb{R}^n\to\mathbb{R},g:\mathbb{R}^n\to\mathbb{R}$ are convex objective functions, $A\in\mathbb{R}^{p\times n},B\in\mathbb{R}^{p\times m}$ and $c\in\mathbb{R}^p$ defines the linear equality constraint. The ADMM algorithm follows a dual ascent procedure with the following augmented Lagrangian:
\begin{equation*}
  \resizebox{1\hsize}{!}{$
  \mathcal{L}_\rho(x,z,y) = f(x) + g(z)+y^\intercal(Ax + Bz - c)+\frac{\rho}{2}\left\|Ax + Bz - c\right\|^2,
  $}
\end{equation*}
where $y$ is the Lagrange multiplier of the equality constraint, $\rho>0$ introduces the quadratic terms that facilitates the convergence of the algorithm. The ADMM updates the follows
\begin{equation}\label{eq:basic_admm}
  \begin{aligned}
  x^{k+1}&\coloneqq \mathop{\arg\min}\nolimits_{x} \mathcal{L}_\rho(x,z^k,y^k)\\
  z^{k+1}&\coloneqq \mathop{\arg\min}\nolimits_{z} \mathcal{L}_\rho(x^{k+1},z,y^k)\\
  y^{k+1}&\coloneqq y^k + \rho(Ax^{k+1} + Bz^{k+1} - c).
  \end{aligned}
\end{equation}
More detail about the ADMM algorithm, including the convergence proof, can be found in \cite{boyd2011distributed}.

\newsec{Multiblock ADMM and convergence.} The standard ADMM handles the case with $N=2$ coupled variables, the cases where $N\ge3$ is often referred to as multiblock ADMM, which solves the following problem:
\begin{equation}\label{eq:multiblock_admm}
  \begin{aligned}
  \mathop {\min }\limits_{\mathbf{x}_i} &\sum\nolimits_{i=1}^N F_i(\mathbf{x}_i)  \\
  \mathrm{s.t.}&\;\sum\nolimits_{i=1}^N A_i \mathbf{x}_i - c = 0.
\end{aligned}
\end{equation}
There are several different setups for multiblock ADMM. A direct extension of the standard ADMM algorithm to the multiblock case is called a Gauss-Seidel ADMM. The Lagrangian for multiblock ADMM can be directly extended from the standard case:
\begin{equation*}
  \resizebox{0.9\hsize}{!}{$
  \mathcal{L}_\rho(\mathbf{x}_1,...\mathbf{x}_N,y) = \sum\limits_{i=1}^N f_i(\mathbf{x}_i)+y^\intercal(\sum\limits_{i=1}^N A_i \mathbf{x}_i-c)+\frac{\rho}{2}\left\|\sum\limits_{i=1}^N A_i \mathbf{x}_i-c\right\|^2.
  $}
\end{equation*}

The Gauss-Seidel iteration is then
\begin{equation}\label{eq:GS_admm}
  \begin{aligned}
  \mathbf{x}_1^{k+1}&\coloneqq \mathop{\arg\min}\nolimits_{\mathbf{x}_1} \mathcal{L}_\rho(\mathbf{x}_1,\mathbf{x}_2^k,...,\mathbf{x}_N^k,y^k)\\
  \mathbf{x}_2^{k+1}&\coloneqq \mathop{\arg\min}\nolimits_{\mathbf{x}_2} \mathcal{L}_\rho(\mathbf{x}_1^{k+1},\mathbf{x}_2,\mathbf{x}_3^k,...,\mathbf{x}_N^k,y^k)\\
  &\quad\quad\quad\cdots\\
  \mathbf{x}_N^{k+1}&\coloneqq \mathop{\arg\min}\nolimits_{\mathbf{x}_N} \mathcal{L}_\rho(\mathbf{x}_1^{k+1},...\mathbf{x}_{N-1}^{k+1},\mathbf{x}_N,y^k)\\
  y^{k+1}&\coloneqq y^k + \rho(\sum\limits_{i=1}^N A_i \mathbf{x}^{k+1}_i - c).
  \end{aligned}
\end{equation}
Note the update of $\mathbf{x}_i$ leverages the updated value of $\mathbf{x}_1,...\mathbf{x}_{i-1}$, which helps the convergence of the algorithm due to the updated information. See \cite{lin2015global} for convergence analysis.

However, as is clear from the algorithm, Gauss-Seidel ADMM cannot be parallelized and is not suitable for a fully decentralized setting. Jacobian ADMM refers to the ADMM algorithms that simultaneously update all $x_i$ with the value from the last iteration, which can be parallelized. However, due to the lack of updated information, Jacobian ADMM is less likely to converge comparing to Gauss-Seidel ADMM. To facilitate convergence, the common strategy is to damp the changing of the variables between iterations, such as adding an underrelaxation \cite{he2015full} and adding a proximal term to the Lagrangian \cite{deng2017parallel}. We shall adopt the latter, named Prox-JADMM, as our ADMM strategy. Prox-JADMM adds a proximal term to primal update, and the ADMM update follows the following procedure:
\begin{equation}\label{eq:prox_admm}
\resizebox{0.9\hsize}{!}{$
  \begin{aligned}
  \mathbf{x}_1^{k+1}&\coloneqq \mathop{\arg\min}\nolimits_{\mathbf{x}_1} \mathcal{L}_\rho(\mathbf{x}_1,\mathbf{x}_2^k,...,\mathbf{x}_N^k,y^k) + \left\|\mathbf{x}_1-\mathbf{x}_1^k\right\|_{P_1}^2\\
  \mathbf{x}_2^{k+1}&\coloneqq \mathop{\arg\min}\nolimits_{\mathbf{x}_2} \mathcal{L}_\rho(\mathbf{x}_1^k,\mathbf{x}_2,\mathbf{x}_3^k,...,\mathbf{x}_N^k,y^k)+\left\|\mathbf{x}_2-\mathbf{x}_2^k\right\|_{P_2}^2\\
  &\quad\quad\quad\cdots\\
  \mathbf{x}_N^{k+1}&\coloneqq \mathop{\arg\min}\nolimits_{\mathbf{x}_N} \mathcal{L}_\rho(\mathbf{x}_1^k,...\mathbf{x}_{N-1}^k,\mathbf{x}_N,y^k)+\left\|\mathbf{x}_N-\mathbf{x}_N^k\right\|_{P_N}^2\\
  y^{k+1}&\coloneqq y^k + \gamma\rho(\sum\limits_{i=1}^N A_i \mathbf{x}^{k+1}_i - c),
  \end{aligned}
  $}
\end{equation}
where $\gamma>0$, and $\left\|\cdot\right\|_{P_i}$ is the 2-norm induced by $P_i\succeq 0$.

It is shown in \cite{deng2017parallel} that when $\rho,\gamma$ and $\{P_i\}$ satisfies
\begin{equation}\label{eq:conv_cond}
  \left\{ {\begin{array}{*{20}{l}}
  {P_i\succ \rho(\frac{1}{\epsilon_i}-1)A_i^\intercal A_i},~~i=1,2,...,N \\
  {\sum\nolimits_{i=1}^N\epsilon_i\le 2-\gamma},
\end{array}} \right.
\end{equation}
for some $\epsilon_i>0$, the Prox-JADMM converges with rate $o(1/k)$. However, it is obvious that \eqref{eq:conv_cond} becomes increasingly difficult to satisfy as $N$ increases. In a fully decentralized problem such as autonomous driving, there may be millions of agents (considering all vehicles in the road system) and typically no agent knows the total number of agents. Fortunately, this convergence proof can be very loose and in practice, we found that the convergence of Prox-JADMM is much better than what the theory predicts, and the convergence rate does not vary much as the number of agents grows.
\begin{figure}
  \centering
  \includegraphics[width=0.9\columnwidth]{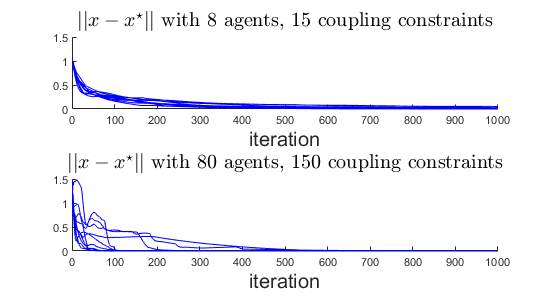}
  \caption{Convergence rate with 8 and 80 agents in 10 trials}\label{fig:convergence}
  \vspace{-0.6cm}
\end{figure}
Fig. \ref{fig:convergence} shows the convergence in 10 randomly generated trials with 8 and 80 agents, and the rate of convergence is similar.

\begin{rem}
  The proposed framework allows for a fully distributed implementation with a synchronized dual update, which can be achieved with a synchronized clock and time stamps. This is due to the fact that the coupling equality constraints are local.
\end{rem}

\section{Prox-JADMM for online decentralized decision-making}\label{sec:main}
In this section, we discuss the Prox-JADMM for online decentralized decision-making.
\vspace{-0.5cm}
\subsection{Quality of unconverged solution}
The optimization problem shown in \eqref{eq:relaxed_setup} falls into the multiblock ADMM framework, and as discussed in Section \ref{sec:ADMM}, in a fully decentralized setting, we use the Prox-JADMM introduced in \cite{deng2017parallel}. Under an online setting, Prox-JADMM may not converge in time, and the unconverged solution is then taken as the solution to the decision making problem after a final correction step, which simply solves the primal update with the current $y$ and without the proximal term:
\begin{equation}\label{eq:correction}
 \resizebox{0.87\hsize}{!}{$
 \forall i=1,...,N, \mathbf{x}_i= \mathop{\arg\min}\limits_{\mathbf{x}_i} \mathcal{L}_\rho(\mathbf{x}_1^k,,...,\mathbf{x}_{i-1}^k,\mathbf{x}_i,\mathbf{x}_{i+1}^k,...,\mathbf{x}_N^k,y^k)
 $}
\end{equation}
where $k$ is the last iteration.

In the setup presented in Section \ref{sec:setup}, each $\mathbf{x}_i$ contains $x_i$ and the local copies $x_j^i$ of its neighbors, $x_i$ is then taken as the local decision. Let $x\coloneqq[x_1^\intercal,...x_N^\intercal]^\intercal$ and let $y_{ij}$ be the Lagrange multiplier associated with the equality constraint $x_j^i=x_j$, we have the following theorem that quantifies the optimality gap.

\begin{thm}\label{thm:quality}
  Given any unconverged solution $\{\mathbf{x}_i\}$ from Prox-JADMM after the correction in \eqref{eq:correction}, the objective of \eqref{eq:relaxed_setup} $J(\mathbf{x})$ and the optimal solution $J^\star(\mathbf{x}^\star)$ satisfies
  \begin{equation}\label{eq:gap}
  \resizebox{0.87\hsize}{!}{$
  \begin{aligned}
    J(\mathbf{x})-J(\mathbf{x}^\star)\le &\sum\limits_{(s,A_s,b_s)\in\mathcal{C}}\frac{\beta}{|s|}\sum\limits_{i\in s}\sum\limits_{j\in s,j\neq i} |A_s^j|\cdot|x_j^i-x_j| \\
    &-\sum\limits_{i}\sum\limits_{j\in\mathcal{N}_i}\left( (y^k_{ij})^\intercal (x_j^i-x_j)+\frac{\rho}{2}||x_j^i-x_j||^2\right).
    \end{aligned}
    $}
  \end{equation}
\end{thm}
\begin{proof}
  The gap can be split into two parts, $J(\mathbf{x})-\sum_i F_i(\mathbf{x}_i)$ and $\sum_i F_i(\mathbf{x}_i)-J(x^\star)$.  Let $\bar{\mathbf{x}}_i\coloneqq [x_i^\intercal,x_{j_1}^\intercal,...,x_{j_{N_i}}^\intercal]^\intercal$ be the vector of the true decision variables (as opposed to the local copies) of agent $i$ and its neighbors, by \eqref{eq:aug_obj}, we have
  \begin{equation*}
  \resizebox{1\hsize}{!}{$
  \begin{aligned}
    &J(\mathbf{x})=\sum\limits_i F_i(\bar{\mathbf{x}}_i),\\
    &\sum\limits_i F_i(\bar{\mathbf{x}}_i)-\sum\limits_i F_i(\mathbf{x}_i)\le \sum\limits_{(s,A_s,b_s)\in\mathcal{C}}\frac{\beta}{|s|}\sum\limits_{i\in s}\sum\limits_{j\in s,j\neq i} |A_s^j|\cdot|x_j^i-x_j|.
    \end{aligned}
    $}
  \end{equation*}
This gives rise to the first line of \eqref{eq:gap}.
Since $\mathbf{x}_i$s are the optimal solutions to the correction step in \eqref{eq:correction}, we have $\mathcal{L}(\mathbf{x}_1,...,\mathbf{x}_N,y^k)\le \mathcal{L}(\mathbf{x}^\star_1,...,\mathbf{x}^\star_N,y^k)=J(\mathbf{x}^\star),$
implying
\begin{equation*}
\resizebox{1\hsize}{!}{$\sum_i F_i(\mathbf{x}_i)-J(x^\star)\le -\sum\limits_{i}\sum\limits_{j\in\mathcal{N}_i}\left( (y^k_{ij})^\intercal (x_j^i-x_j)+\frac{\rho}{2}||x_j^i-x_j||^2\right).$}
\end{equation*}
Combining the two parts proves the bound.
\end{proof}
Theorem \ref{thm:quality} shows that the mismatch $x_j^i-x_j$ is an important indicator of the convergence of the algorithm. When the mismatch between local copies and the true values are small, the optimality gap is small.

\subsection{Online decision-making with Prox-JADMM}
Online decision-making is very common in engineering systems. Our motivating example is a control barrier function quadratic programming (QP), which is determined by the state of the agents and needs to be solved for every time step. Similar problems include the optimal power flow (OPF) problem for power systems and decentralized MPC. To distinguish the sequential decision making from the iterations of ADMM between time steps, we let $f_i[t]$, $\mathcal{C}[t]$ denote the local objective functions and the constraint set at time $t$. The ADMM algorithm would take multiple iterations to obtain a solution within one time step, then solve an updated problem with $f_i[t+1]$ and $\mathcal{C}[t+1]$. Typically the decision making problem satisfies some continuity condition, i.e., the problem would gradually change over time. To be more precise, let
 \begin{equation}\label{eq:F}
  F(x)=\sum\nolimits_{i} {f_i(x_i)}+\beta\mathds{1}^\intercal \max\{\mathbf{0},Ax-b\},
 \end{equation}
 be defined as the total objective on $\mathcal{X}\coloneqq \mathcal{X}_1\times...\times\mathcal{X}_N$. We begin with the following assumptions.
\begin{ass}\label{ass:continuity}
  The change of the local objective function $f_i$ and the constraint set $\mathcal{C}$ over successive time steps is bounded such that there exists $\kappa>0$ that $F[t+1]-F[t]$ is Lipschitz continuous with constant $\kappa$ within $\mathcal{X}$, i.e.,
\begin{equation*}
\begin{aligned}
&\forall t=0,1,...,~\forall x,x'\in \mathcal{X}, \\
&|(F[t+1]-F[t])(x')-(F[t+1]-F[t])(x)|\le\kappa ||x-x'||.
\end{aligned}
\end{equation*}
\end{ass}
\begin{ass}\label{ass:strong_conv}
  For all agents, the local objective function $f_i$ is twice differentiable and strongly convex with parameter $\sigma$, and $\nabla f_i$ is Lipschitz with constant $\nu$.
\end{ass}

\begin{prop}\label{prop:continuity}
  Under Assumptions \ref{ass:continuity}, \ref{ass:strong_conv}, let $x^\star[t]$ be the optimal solution to \eqref{eq:relaxed_setup} at time $t$, then $||x^\star[t+1]-x^\star[t]||\le \sigma^{-1} \kappa$.
\end{prop}

\begin{proof}
  Proposition \ref{prop:continuity} is directly adopted from Proposition 4.32 from \cite{bonnans2013perturbation}. Since $f_i$s are strongly convex with parameter $\sigma$, $F$ is also strongly convex with parameter $\sigma$, thus satisfies the second order growth condition with parameter $\sigma$. Then the result follows by taking $\epsilon$ to 0.
\end{proof}

Proposition \ref{prop:continuity} gives continuity condition for the optimal solution of \eqref{eq:relaxed_setup}, but the dual variable is also critical to the convergence of ADMM. Although strong duality guarantees the uniqueness of $x^\star$, the dual variables that satisfy the Karush-Kuhn-Tucker (KKT) condition may not be unique, which is mainly due to the nondifferentiability of the $\max$ function. One solution is to use a smooth approximation of the $\max$ function, for example, with the \textit{LogSumExp} function. The smooth approximation of $F_i$ is defined as
\begin{equation}\label{eq:aug_obj_smooth}
\bar{F}_i(\mathbf{x}_i)= {f_i(x_i)}+\frac{\beta}{c}(w^i)^\intercal \log (e^{c(\mathbf{A}^i\mathbf{x}_i-\mathbf{b}^i)}+\mathds{1}),
\end{equation}
where $\exp$ and $\log$ are taken entrywise and $c>0$ is a parameter that tunes the smoothness of the approximation. The larger $c$ is, the closer the approximation is to the original function. The total objective after smoothing is then
\begin{equation}\label{eq:total_obj_smooth}
  \bar{F}(x)=\sum\nolimits_{i=1}^N f_i(x_i)+\frac{\beta}{c}\mathds{1}^\intercal \log(e^{c(Ax-b)}+\mathds{1}).
\end{equation}

\begin{prop}\label{prop:smooth_continuity}
  Under Assumptions \ref{ass:strong_conv} and assume $\bar{F}[t+1]-\bar{F}[t]$ is Lipschitz with constant $\kappa$ for all $t$, let $x^\star[t],y^\star[t]$ be the solution of \eqref{eq:relaxed_setup} with $F_i$ replaced by $\bar{F}_i$ defined in \eqref{eq:aug_obj_smooth} at time $t$, let $|s|^{\max}$ be the maximum number of agents coupled by a single constraint in $\mathcal{C}$, then $||x^\star[t+1]-x^\star[t]||\le \sigma^{-1} \kappa$, and $||y^\star[t+1]-y^\star[t]||<||A||_F^2 \frac{|s|^{\max}-1}{|s|^{\max}}\beta c \sigma^{-1} \kappa$.
\end{prop}
\begin{proof}
Following a similar argument made in the proof of Proposition \ref{prop:continuity}, $||x^\star[t+1]-x^\star[t]||\le \sigma^{-1} \kappa$. The dual variable corresponds to the equality constraint such that all the local copies are equal to the actual variable. Let $y_{ij}$ be the dual variable of the constraint $x_j^i=x_j$, then $y^\star_{ij}[t]$ satisfies $
    y^\star_{ij}[t]= -\frac{\partial\bar{F}_i(\mathbf{x}_i^\star[t])}{\partial x_j^i}$.
  The Hessian of the partial derivative is $\frac{\partial^2 \bar{F}_i}{\partial (x_j^i)^2}=\sum_{n=1}^{m_i}\frac{e^{\mathbf{A}_n^i\mathbf{x}_i-\mathbf{b}_n^i}\beta c}{(e^{\mathbf{A}_n^i\mathbf{x}_i-\mathbf{b}_n^i}+1)^2 |s_n|} (\mathbf{A}^i_{n,j})^\intercal \mathbf{A}^i_{n,j} $, where $\mathbf{A}^i_n$ is the $n$th row of $\mathbf{A}^i$, and $\mathbf{A}_{n,j}^i$ is the block corresponding to $x_j^i$. This implies that $\frac{\partial\bar{F}_i}{\partial x_j^i}$ has a Lipschitz constant less than
  $\sum_{n=1}^{m_i}\frac{\beta c}{ |s_n|} ||\mathbf{A}^i_{n,j}||_F^2$, i.e., $||y^\star_{ij}[t-1]-y^\star_{ij}[t]||\le \sum_{n=1}^{m_i}\frac{\beta c}{ |s_n|} ||\mathbf{A}^i_{n,j}||_F^2||x_j^\star[t+1]-x_j^\star[t]||$. Sum up for all $i,j$,
  \begin{equation*}
  \resizebox{1\hsize}{!}{$
    ||y^\star[t+1]-y^\star[t]||< \sum\limits_{(s,A_s,b_s)\in\mathcal{C}}\frac{(|s|-1)\beta c}{|s|}\sum_{i\in s} ||A^i_s||_F^2 ||x_i^\star[t+1]-x_i^\star[t]||.
    $}
  \end{equation*}
  Using Cauchy-Schwartz inequality, we get

  \resizebox{0.88\hsize}{!}{$
    \begin{aligned}
    ||y^\star[t+1]-y^\star[t]||&<||A||_F^2 \frac{|s|^{\max}-1}{|s|^{\max}}\beta c||x^\star[t+1]-x^\star[t]||\\
    &\le||A||_F^2 \frac{|s|^{\max}-1}{|s|^{\max}}\beta c \sigma^{-1} \kappa.
    \end{aligned}
    $}
\end{proof}

\begin{figure}
  \centering
  \includegraphics[width=0.9\columnwidth]{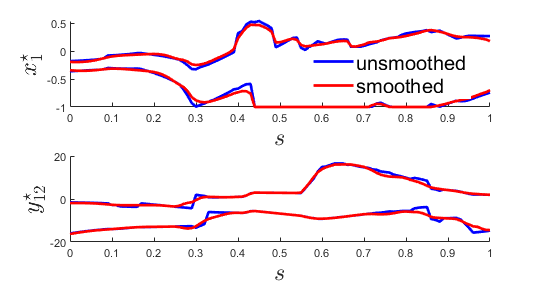}
  \caption{Optimal primal and dual values with and without \textit{LogSumExp} smoothing}\label{fig:smoothing}
  \vspace{-0.5cm}
\end{figure}

In practice, we found that smoothing may be unnecessary if the dual variable in the original problem is smooth enough w.r.t. the problem data. Fig. \ref{fig:smoothing} shows the result of a numeric experiment with 8 agents, each $x_i\in\mathbb{R}^2$, $f_i$ quadratic, and the coupling constraints are pairwise. We let the coupling constraint $Ax\le b$ gradually change from the initial value to the final value via linear interpolation: $A=\lambda A^1+(1-\lambda)A^2,b=\lambda b^1+(1-\lambda)b^2$, and let $\lambda$ change from 0 to 1, where $A^1$ and $A^2$ share the same sparsity pattern. Fig. \ref{fig:smoothing} shows the optimal value of $x_1$ and $y_{12}$ as $\lambda$ changes. The smoothing made both $x^\star$ and $y^\star$ smoother, yet the unsmoothed curve is already continuous with a reasonably small Lipschitz constant.
\begin{rem}
The disadvantage of smoothing in this case is that it makes the local optimization from a quadratic programming to an optimization with the \textit{LogSumExp} function, which is still convex, but there lacks specialized solvers for it. Therefore, when $y$ is smooth enough w.r.t. the problem data without smoothing, we choose to not use smoothing.
\end{rem}

The sensitivity analysis above shows that when the problem data $\{f_i\}$ and $\mathcal{C}$ change gradually over time, the optimal value also changes continuously. This motivate us to warm start the ADMM iteration with the solution of both the primal and dual variables (though unconverged) from last the iteration. In \cite{deng2017parallel}, convergence was proved by showing the distance from the current solution to the optimal solution is monotonically decreasing. Although no quantitative analysis of the convergence rate was given, one would expect that if the change of the optimal solution between successive time steps is bounded, the convergence of ADMM can catch up and maintain a certain maximum distance to the optimal solution, which is demonstrated in Section \ref{sec:result}.

\begin{algorithm}[b]
    \caption{Warm start online Prox-JADMM}
    \label{alg:seq_ADMM}
    \begin{algorithmic}[1] 
        \State $t\gets 0$, $k\gets 0$
        \State $\{\mathbf{x}^0_i\}\gets \mathbf{0}$, $y^0\gets \mathbf{0}$
        \While {\textbf{NOT} TERMINATE}
            \State Obtain $\{F_i\}[t]$, calculate $\mathcal{L}$
            \For {$k=0:M$}
                \State Prox-JADMM update following \eqref{eq:prox_admm}
            \EndFor
            \State Correction step following \eqref{eq:correction}
            \State $t\gets t+1$, $k\gets 0$, $\{\mathbf{x}^0_i\}\gets \{\mathbf{x}^M_i\}$, $y^0\gets y^M$
        \EndWhile
    \end{algorithmic}
\end{algorithm}
The warm start online Prox-JADMM is summarized in Algorithm \ref{alg:seq_ADMM}. $M$ is the maximum number of iterations allowed between time steps (determined by the computation speed and communication speed). At the beginning of each time step, the primal and dual variables are warm started with the (possibly unconverged) values from the last time step.
\vspace{-0.3cm}
\section{Application to multiagent CBF}\label{sec:result}
The motivating example is a multiagent control barrier function (CBF) problem that arises in the control of autonomous vehicles. Control barrier function \cite{ames2016control} considers a control affine dynamic system described by a differential equation
\begin{equation}\label{eq:dyn}
  \dot{x}=f(x)+g(x)u, x\in\mathbb{R}^n,u\in\mathcal{U}\subseteq \mathbb{R}^m,
\end{equation}
where $x$ is the state, subject to a safety constraint $x\notin\mathcal{X}_d$, $u$ is the input.
\begin{rem}
To avoid confusion, $x$ denotes system state, following the convention of control literature, the input $u$ is the decision variable of the ADMM problem introduced later.
\end{rem}
A CBF is then a function $h:\mathbb{R}^n\to\mathbb{R}$ that satisfies
\begin{equation}\label{eq:CBF}
  \begin{aligned}
&\forall~ x \in {\mathcal{X}_0},&h(x) \ge 0\\
&\forall~ x \in {\mathcal{X}_d},&h(x) < 0\\
&\forall~ x \in \left\{ {x\mid h(x) \ge 0} \right\}, &\exists~ u \in \mathcal{U}\;
\mathrm{s.t.}~\dot h + \alpha \left( h \right) \ge 0,
\end{aligned}
\end{equation}
where $\mathcal{X}_0$ is the set of initial states, $\alpha(\cdot)$ is a class-$\mathcal{K}$ function, i.e., $\alpha(\cdot)$ is strictly increasing and satisfies $\alpha(0)=0$. The last condition is called the CBF condition, which guarantees that any state within the safe set ($h(x)\ge 0$) remains safe, and any state outside the safe set converges back to the safe set exponentially. Given a legacy controller $u^0:\mathbb{R}^n\to\mathcal{U}$, the following CBF QP solves for the minimum intervention over $u^0(x)$ while enforcing the CBF condition:
\begin{equation}\label{eq:CBF_QP}
  \begin{aligned}
u^\star = \mathop {\arg \min }\limits_{u \in \mathcal{U}} &\left\| {u - {u^0}(x)} \right\|^2\\
\mathrm{s.t.}~~&\nabla h\cdot (f(x)+g(x)u) + \alpha \left( h \right) \ge 0,
\end{aligned}
\end{equation}
which is a QP subject to a linear inequality constraint. The feasibility of \eqref{eq:CBF_QP} when $h(x)\ge 0$ is guaranteed if the set $\{x|h(x)\ge 0\}$ is a control invariant set. This is sufficient to guarantee safety, since under \eqref{eq:CBF_QP}, any nonnegative $h$ will remain nonnegative. In particular, \cite{chen2020guaranteed} presents a CBF based on backup policies, which is applicable to multiagent problems and the feasibility of the decentralized CBF QP is guaranteed when $h\ge 0$. To be specific, consider $N$ agents with dynamics $x_i=f_i(x_i)+g_i(x_i)u_i$, where $x_i$ and $u_i$ are the state and input of agent $i$. The CBF for multiagent collision avoidance can be decomposed as follows:
\begin{equation}\label{eq:CBF_decom}
  h(x_1,...x_N)=\mathop{\min}\limits_{i=1,..,N} \left[h_i(x_i),\mathop{\min}\limits_{j\in\mathcal{N}_i}h_{ij}(x_i,x_j)\right],
\end{equation}
where $h_i$ is the local CBF for agent $i$, $h_{ij}$ is the CBF for obstacle avoidance between agent $i$ and $j$, defined for all of agent $i$'s neighbors $\mathcal{N}_i$. The multiagent CBF QP is then
\begin{equation}\label{eq:cen_CBF}
\resizebox{1\hsize}{!}{$
  \begin{aligned}
 \mathop { \min }\limits_{u_1,..,u_N} &\sum_i\left\| {u_i - {u^0_i}} \right\|^2\\
\mathrm{s.t.}~~&\forall i, \nabla h_i\cdot f_i( {x_i,u_i} ) + \alpha \left( h_i \right) \ge 0,\\
&\forall j\in\mathcal{N}_i,~\nabla h_{ij}\cdot\left( f_i( x_i)+g(x_i)u_i+f_j(x_j)+g(x_j)u_j \right) + \alpha \left( h_{ij} \right) \ge 0,
\end{aligned}
$}
\end{equation}
which falls into the form of \eqref{eq:setup}. It is shown in \cite{chen2020guaranteed} that when all the $h_i$'s and $h_{ij}$'s are positive, \eqref{eq:cen_CBF} is always feasible, and a feasible solution can be obtained by a fully decentralized optimization. However, when $h<0$, there is no guarantee of feasibility, and the decentralized algorithm without communication proposed in \cite{chen2020guaranteed} may perform badly. \cite{PCCA} proposed a presumed cooperation algorithm that also uses local copies, but consensus cannot be achieved without communication. This motivates us to use ADMM to improve the performance via consensus. As mentioned in Section \ref{sec:setup}, the goal is to minimize the violation of the CBF condition. The local optimization for agent $i$ is then the following:
\begin{equation}\label{decen_CBF}
\resizebox{1\hsize}{!}{$
   \begin{aligned}
\mathop {\min }\limits_{\mathbf{u}_i} &\left\| {u_i - {u^0_i}} \right\|^2+\beta\max\{0,-\nabla h_i\cdot f_i( {x_i,u_i} ) - \alpha \left( h_i \right)\}\\
&+\frac{\beta}{2}\sum\limits_{j\in\mathcal{N}_i}\max\{\mathbf{0},-\nabla h_{ij}\cdot(f_i(x_i)+g_i(x_i)u_i+f_j(x_j)+g_j(x_j)u^i_j) - \alpha \left( h_{ij} \right)\}\\
\mathrm{s.t.}&~ u_i\in\mathcal{U}_i
\end{aligned}
$}
\end{equation}
where $\mathbf{u}_i$ contains $u_i$ and $u_j^i$, the local copies of the input of all its neighbors. Note that the penalty for violating $h_{ij}$s are multiplied by $\frac{1}{2}$ since every $h_{ij}$ is shared by two agents, thus the penalty is also shared by two agents.

The CBF condition is determined by the state of the agents, changing gradually over time. Therefore, it is straightforward to find the Lipschitz constant that characterizes the continuity of the online decision-making problem. The strong convexity comes from the quadratic local cost of each agent.

\begin{rem}
  Comparing to the cooperative control methods, the Prox-JADMM CBF approach is fully decentralized, the agents do not know the total number of agents, they simply exchange messages with their neighbors and solve the local update in a synchronous and parallel fashion. Moreover, The agents do not need to share their local objective functions.
\end{rem}

We consider the merging case in autonomous driving, in which the CBF may start negative since the vehicles on the two lanes can detect each other only when they are close to the merging point, and their current path may lead to collision. Dubin's car model is used for the vehicle dynamics with acceleration and yaw rate as input. The time step for updating the control input is 50ms, and the maximum iteration $M$ for the Prox-JADMM is 30, chosen based on the sampling time for control update. The CBF QP only includes $h_i$ and $h_{ij}$ that are smaller than a threshold (large $h$ indicates the constraint is satisfied with a big margin), which localizes the problem and bounds the size of the local optimization.
\begin{figure}[t]
  \centering
  \includegraphics[width=1\columnwidth]{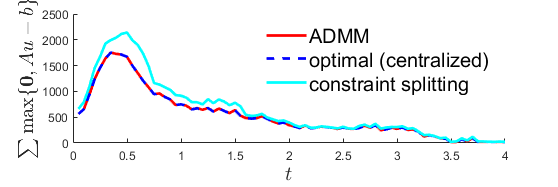}
  \caption{Total constraint violation of the constraint splitting method used in \cite{chen2020guaranteed}, the Prox-JADMM, and the optimal solution solved with centralized optimization}\label{fig:merge_comparison}
  \vspace{-0.5cm}
\end{figure}
Fig. \ref{fig:merge_comparison} shows the performance comparison of the Prox-JADMM algorithm with the benchmark in \cite{chen2020guaranteed} and the optimal solution. The total constraint violation of Prox-JADMM is always smaller than the benchmark, and very close to the optimal solution, as Theorem \ref{thm:quality} predicts.

Fig. \ref{fig:warmstart} shows the effect of warm start where $\|\Delta u\|\doteq\sum_{i,j} ||u_j^i-u_j||$, the violation of the consensus constraint, which is shown to be an important indicator of the convergence of ADMM.  The blue and red curves show the consensus violation at the beginning and at the end (after M iterations) of the ADMM iterations, respectively. The warm start help reducing $\Delta u$ after the update of the optimization problem and ensures good convergence of the ADMM, whereas the convergence can be bad without the warm start.

A video of the highway sim can be found \href{https://youtu.be/9qjHuDkpd2E}{youtu.be/9qjHuDkpd2E} and a python realization of the Prox-JADMM algorithm with linear inequality constraint can be found \href{https://github.com/chenyx09/LCADMM}{github.com/chenyx09/LCADMM}.
\begin{figure}[t]
  \centering
  \includegraphics[width=0.9\columnwidth]{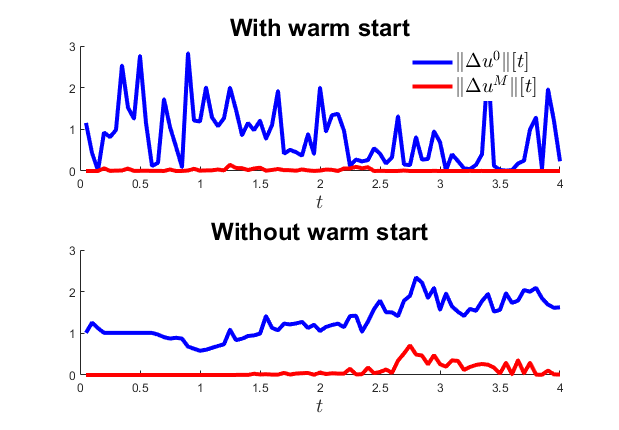}
  \caption{Effect of warm start on convergence}\label{fig:warmstart}
  \vspace{-0.5cm}
\end{figure}
\section{Conclusion}\label{sec:conclusion}
We present a Prox-JADMM algorithm for online decentralized decision making with coupling inequality constraints. We show that under gradual change of the problem data, the optimal solution changes continuously, and warm starting the ADMM helps accelerate convergence between time steps. The algorithm is applied to a decentralized control barrier function problem and the simulation shows that the decentralized algorithm achieves near optimal performance.
\balance
\renewcommand{\baselinestretch}{0.9}
\bibliographystyle{myieeetran}
\bibliography{my_bib}
\end{document}